\documentclass[12pt,a4paper]{article}
\usepackage{fullpage}
\usepackage{amssymb,amsmath,stmaryrd,amsthm}
\usepackage[mathscr]{eucal}
\usepackage{pstricks,pst-node,pst-text,pst-3d}

\def \C {\mathfrak{C}}

\theoremstyle{definition}
\newtheorem{definition}{Definition}

\newtheorem{proposition}{Proposition}

\newtheorem{corollary}{Corollary}

\newtheorem{example}{Example}
\theoremstyle{remark}

\begin{document}

\title{The lattice of embedded subsets} 

\author{Michel GRABISCH\\
\normalsize Universit\'e de Paris I  -- Panth\'eon-Sorbonne \\
\normalsize     Centre d'Economie de la Sorbonne\\ 
\normalsize 106-112 Bd. de l'H\^opital, 75013 Paris,    France\\
\normalsize email \texttt{michel.grabisch@univ-paris1.fr}\\
\normalsize Tel (+33) 1-44-07-82-85, Fax (+33) 1-44-07-83-01}

\date{}

\maketitle

\mbox{}

\begin{abstract}
In cooperative game theory, games in partition function form are real-valued
function on the set of so-called embedded coalitions, that is, pairs $(S,\pi)$
where $S$ is a subset (coalition) of the set $N$ of players, and $\pi$ is a
partition of $N$ containing $S$. Despite the fact that many studies have been
devoted to such games, surprisingly nobody clearly defined a structure (i.e., an
order) on embedded coalitions, resulting in scattered and divergent works,
lacking unification and proper analysis. The aim of the paper is to fill this
gap, thus to study the structure of embedded coalitions (called here embedded
subsets), and the properties of games in partition function form.
\end{abstract}


\section{Introduction}
The Boolean lattice of subsets and the lattice of partitions are two well-known
posets, with numerous applications in decision, game theory,
classification, etc. 

We consider in this paper a more complex structure, which is in some sense a
combination of the above two. Its origin comes from cooperative game theory. Let
us consider a set $N$ of $n$ players, and define a real-valued function $v$ on
$2^N$, such that $v(\emptyset) = 0$. Such a function is called a \emph{game} on
$N$, and for any coalition (subset) $S\subseteq N$, the quantity $v(S)$
represents the ``worth'' or ``power'' of coalition $S$. Hence a game is a
function on the Boolean lattice $2^N$ vanishing at the bottom element. Having a
closer look at the meaning of $v(S)$, we could say more precisely: suppose that
players in $S$ form a coalition, the other players in $N\setminus S$ forming the
``opponent'' group. Then $v(S)$ is the amount of money earned by $S$ (or the power
of $S$) in such a dichotomic situation. 

A more realistic view would be to consider that the opponent group may also be
divided into groups, say $S_2,\dots, S_k$, so that $\{S_2,\ldots,S_k\}$ form a
partition of $N\setminus S$. In this case it is likely that the value earned by
$S$ may depend on the partition of $N\setminus S$. Hence we are lead to define
the quantity $v(S,\pi)$, where $\pi$ is a partition of $N$ containing $S$ as a
block, i.e., $\pi=\{S,S_2,\ldots,S_k\}$.

Such games are called \emph{games in partition function form}, while $(S,\pi)$
is called an \emph{embedded coalition}, and have been introduced by Thrall and
Lucas \cite{thlu63}. Despite the fact that many works have been undertaken on
this topic (let us cite, among others, Myerson \cite{mye77}, Bolger
\cite{bol89}, Do and Norde \cite{dono02}, Fujinaka \cite{fuj05a}, Clippel and
Serrano \cite{clse05}, Albizuri et al. \cite{alarru05}, Macho-Stadler et
al. \cite{mapewe07}, who all propose various definitions and axiomatizations of
the Shapley value, and Funaki and Yamato \cite{fuya99} who deal with the core,
etc.), surprisingly nobody has clearly defined a structure for embedded
coalitions. As a consequence, most of these works have divergent point of views,
lack unification, and do not provide a good mathematical analysis of the concept
of game in partition function form. This paper tends to fill this gap. We
propose a natural structure for embedded coalitions (which we call
\emph{embedded subsets}), which is a lattice, study it, and provide a variety of
results described below. Our analysis will consider in particular the following
points, all motivated by game theory but also commonly considered in the field
of posets:
\begin{enumerate}
\item The number of maximal chains between two given embedded subsets. Many
  concepts in game theory are defined through maximal chains, like the Shapley
  value \cite{sha53} and the core of convex games \cite{sha71}. This is
  addressed in Section~\ref{sec:emsu}, where a thorough analysis of the poset of
  embedded subsets is done. Based on these results, Grabisch and
  Funaki propose in \cite{grfu08} a definition of the Shapley value, different
  from the ones cited above, and having good
  properties. 
\item The M\"obius function on the lattice of embedded subsets. Usually the
  function obtained by the M\"obius inversion on a game is called the
  \emph{M\"obius transform} of this game, or \emph{the dividends}. It is a
  fundamental notion, permitting to express a game in the basis of unanimity
  games. This topic is addressed in Section~\ref{sec:mob}, and is a new
  achievement in the theory of games in partition function form.
\item Particular classes of games, like additive games, super- and submodular
  games (Section~\ref{sec:fun}), $\infty$-monotone games, also called, up to
  some differences, positive games or belief functions, and minitive games
  (Section~\ref{sec:bel}). Additive games, corresponding to valuations on
  lattices, are basic in game theory since they permit to define the core and
  all procedures of sharing. Super- and submodular functions are also very common
  in game theory and combinatorial optimization. $\infty$-monotone games are
  especially important because in the classical Boolean case they have
  nonnegative dividends (M\"obius transform), and they are well-known in
  artificial intelligence and decision theory under the name of belief
  functions. Lastly, minitive games, also called necessity functions in
  artificial intelligence, are particular belief functions. In poset language,
  they are inf-preserving mappings. Our analysis brings many new results,
  establishing the existence of these particular classes of games.
\end{enumerate}

\section{Background on partitions}\label{sec:part}
In this section, we introduce our notation and recall useful results on the geometric
lattice of partitions (see essentially Aigner \cite{aig79}, and \cite{gest95}),
and prove some results needed in the following.

We consider the set $[n]=:N$, and denote its subsets by $S,T,S', T',\ldots, $
and by $s,t,s', t',\ldots$ their respective cardinalities. The set of partitions
of $[n]$ is denoted by $\Pi(n)$. Partitions are denoted by $\pi,\pi',\ldots$,
and $\pi=\{S_1,\ldots,S_k\}$, $S_1,\ldots,S_k\in2^N$.  Subsets $S_1,\ldots,S_k$
are called \emph{blocks} of $\pi$. A partition into $k$ blocks is a
$k$-partition. 

Taking $\pi,\pi'$ partitions in $\Pi(n)$, we say that $\pi$ is a
\emph{refinement} of $\pi'$ (or $\pi'$ is a \emph{coarsening} of $\pi$), denoted
by $\pi\leq \pi'$, if any block of $\pi$ is contained in a block of $\pi'$ (or
every block of $\pi'$ fully decomposes into blocks of $\pi$). When endowed with
the refinement relation, $(\Pi(n),\leq)$ is a lattice, called the
\emph{partition lattice}.

We use the following shorthands:
$\pi^\top:=\{N\}$, $\pi^\bot:=\{\{1\},\ldots,\{n\}\}$. For any $\emptyset\neq
S\subset N$, $\pi^\top_S:=\{S,N\setminus S\}$,
$\pi^\bot_S:=\{S,\{i_1\},\ldots,\{i_{n-s}\}\}$, with $N\setminus
S=:\{i_1,\ldots,i_{n-s}\}$. Also, for any two partitions $\pi,\pi'$ such that
$\pi\leq\pi'$, the notation $[\pi,\pi']$ means as usual the set of all
partitions $\pi''$ such that $\pi\leq\pi''\leq\pi'$. 

The following facts on $\Pi(n)$ will be useful in the following:
\begin{enumerate}
\item The number of partitions of $k$ blocks is $S_{n,k}$ (Stirling number of
  the second kind), with 
\[
S_{n,k} := \frac{1}{k!}\sum_{i=0}^n(-1)^{k-i}\binom{k}{i}i^n, \quad n\geq
0,k\leq n.
\]
\item Each partition $\pi$ covers $\sum_{S\in \pi}2^{|S|-1}-|\pi|$ partitions. Each
  $k$-partition is covered by $\binom{k}{2}$ partitions.
\item Let $\pi:=\{S_1,\ldots,S_k\}$ be a $k$-partition. Then we have the
following isomorphisms:
\begin{align*}
[\pi,\pi^\top] & \cong \Pi(k)\\
[\pi^\bot,\pi] & \cong \prod_{i=1}^k\Pi(s_i)\\
[\pi,\pi'] & \cong \prod_{i=1}^{|\pi'|}\Pi(m_i) \text{ for some $m_i$'s with }
\sum_{i=1}^{|\pi'|}m_i = k.
\end{align*}
\end{enumerate}

We give the following results (up to our knowledge, some of them have  not yet been
investigated), which will be used in the following. We use the notation
$\mathcal{C}(P)$ to denote the set of maximal chains from bottom to top in the
poset $P$, whenever this makes sense.
\begin{proposition}\label{prop:1}
Let $\pi,\pi'\in \Pi(n)$ such that $\pi'<\pi$, with $\pi:=\{S_1,\ldots,S_k\}$
and $\pi':=\{S_{11},\ldots,S_{1l_1},S_{21},\ldots,S_{2l_2},\ldots,S_{kl_k}\}$, with
$\{S_{i1},\ldots,S_{il_i}\}$ a partition of $S_i$, $i=1,\ldots,k$, and
$k':=\sum_{i=1}^k l_i$. 
\begin{enumerate}
\item The number of maximal chains of $\Pi(n)$ from bottom to top is 
\[
|\mathcal{C}(\Pi(n))| = \frac{n!(n-1)!}{2^{n-1}}.
\]
\item The number of maximal
chains from $\pi^\bot$ to $\pi$ is
\[
|\mathcal{C}([\pi^\bot,\pi])| = \frac{(n-k)!}{2^{n-k}}s_1!s_2!\cdots s_k!.
\]
\item The number of maximal chains from $\pi$ to $\pi^\top$ is 
\[
|\mathcal{C}([\pi,\pi^\top])|=|\mathcal{C}(\Pi(k))|. 
\]
\item The number of maximal chains from $\pi'$ to $\pi$ is
\[
|\mathcal{C}([\pi',\pi])| = \frac{(k'-k)!}{2^{k'-k}}l_1!l_2!\cdots l_k!
\]
\end{enumerate}
\end{proposition}
\begin{proof}
\begin{enumerate}
\item See Barbut and Monjardet \cite[p. 103]{bamo70}.
\item We use the fact that for any $k$-partition $\pi=\{S_1,\ldots,S_k\}$,
$[\pi^\bot,\pi] = \prod_{i=1}^k \Pi(s_i)$. We have the general following fact:
if $L=L_1\times\cdots\times L_k$, then to obtain all maximal chains in $L$, we
select a $k$-uple of maximal chains $C_1,\ldots,C_k$ in $L_1,\ldots,L_k$
respectively, say of length $c_1,\ldots,c_k$.  Then the number of chains induced
by $C_1,\ldots,C_k$ in $L$ is equal to the number of chains in the lattice
$C_1\times\cdots\times C_k$, isomorphic to the lattice $c_1\times\cdots\times
c_k$ ($c_i$ denotes the linear lattice of $c_i$ elements).  This is
known to be 
\[
\frac{(\sum_{i=1}^{k}c_i)!}{\prod_{i=1}^k( c_i!)}.
\]  
Applied to our case, this gives
\[
|\mathcal{C}([\pi^\bot,\pi])| =
 \prod_{i=1}^k|\mathcal{C}(\Pi(s_i)|\frac{\big(\sum_{i=1}^k
 (s_i-1)\big)!}{\prod_{i=1}^k(s_i-1)!}
\]
which after simplification gives the desired result, using the fact that
$\sum_{i=1}^k(s_i-1)=n-k$. 
\item Immediate from $[\pi,\pi^{\top}]\cong\Pi(k)$.
\item Simply consider
  $S_{11},\ldots,S_{1l_1},S_{21},\ldots,S_{2l_2},\ldots,S_{kl_k}$, and use (ii) for $\Pi(k')$.
\end{enumerate}
\end{proof}

\section{The structure of embedded subsets}\label{sec:emsu}
An \emph{embedded subset} is a pair $(S,\pi)$ where $S\in 2^N\setminus
\{\emptyset\}$ and $\pi\ni S$, where $\pi\in\Pi(n)$.
We denote by
$\mathfrak{C}(N)$ (or by $\mathfrak{C}(n)$) the set of embedded coalitions on $N$. For
the sake of concision, we often denote by $S\pi$ the embedded coalition
$(S,\pi)$, and omit braces and commas for subsets (example with $n=3$:
$12\{12,3\}$ instead of $(\{1,2\},\{\{1,2\},\{3\}\})$). Remark that
$\mathfrak{C}(N)$ is a proper subset of $2^N\times \Pi(N)$.

As mentionned in the introduction, works on games in partition function form do not
explicitly define a structure (that is, some order) on embedded coalitions. A
natural choice is to take the product order on  $2^N\times \Pi(N)$:
\[
(S,\pi) \sqsubseteq (S',\pi') \Leftrightarrow S\subseteq S'\text{ and } \pi\leq\pi'.
\]
Evidently, the top element of this ordered set is $(N,\pi^\top)$ (denoted more
simply by $N\{N\}$ according to our conventions). However, due to the fact that
the empty set is not allowed in $(S,\pi)$, there is no bottom element in the
poset $(\mathfrak{C}(N),\sqsubseteq)$, since all elements of the
form $(\{i\},\pi^\bot)$ are minimal elements. For mathematical convenience, we
introduce an artificial bottom element $\bot$ to $\mathfrak{C}(N)$ (it could be
considered as $(\emptyset,\pi^\bot)$), and denote
$\mathfrak{C}(N)_\bot:=\mathfrak{C}(N)\cup\{\bot\}$. We give as illustration the
partially ordered set $(\mathfrak{C}(N)_\bot,\sqsubseteq)$ with $n=3$
(Fig. \ref{fig:3}).
\begin{figure}[htb]
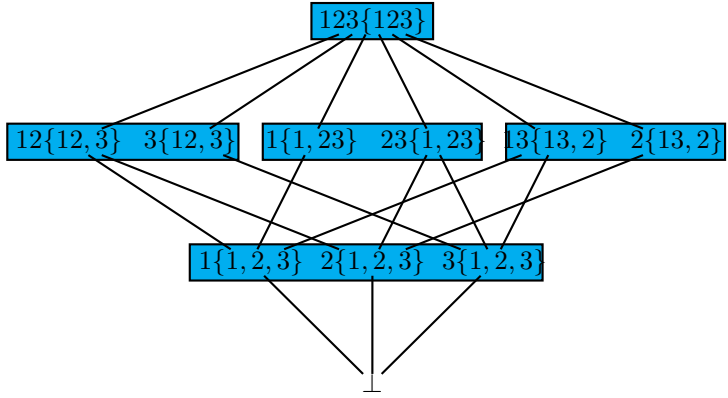

\begin{center}
\psset{unit=0.8cm}
\pspicture(-6,-1)(6,6)
\pspolygon[fillstyle=solid,fillcolor=cyan](-1,5.7)(-1,6.3)(1,6.3)(1,5.7)
\pspolygon[fillstyle=solid,fillcolor=cyan](-6,3.7)(-6,4.3)(-2.2,4.3)(-2.2,3.7)
\pspolygon[fillstyle=solid,fillcolor=cyan](-1.8,3.7)(-1.8,4.3)(1.8,4.3)(1.8,3.7)
\pspolygon[fillstyle=solid,fillcolor=cyan](2.2,3.7)(2.2,4.3)(5.8,4.3)(5.8,3.7)
\pspolygon[fillstyle=solid,fillcolor=cyan](-3,1.7)(-3,2.3)(2.8,2.3)(2.8,1.7)
\rput(0,0){\rnode{0}{\footnotesize $\bot$}}
\rput(-2,2){\rnode{1a}{\footnotesize $1\{1,2,3\}$}}
\rput(0,2){\rnode{2a}{\footnotesize $2\{1,2,3\}$}}
\rput(2,2){\rnode{3a}{\footnotesize $3\{1,2,3\}$}}
\rput(-5,4){\rnode{12b}{\footnotesize $12\{12,3\}$}}
\rput(-3,4){\rnode{3b}{\footnotesize $3\{12,3\}$}}
\rput(-1,4){\rnode{1b}{\footnotesize $1\{1,23\}$}}
\rput(1,4){\rnode{23b}{\footnotesize $23\{1,23\}$}}
\rput(3,4){\rnode{13b}{\footnotesize $13\{13,2\}$}}
\rput(5,4){\rnode{2b}{\footnotesize $2\{13,2\}$}}
\rput(0,6){\rnode{123c}{\footnotesize $123\{123\}$}}
\ncline{0}{1a}
\ncline{0}{2a}
\ncline{0}{3a}
\ncline{1a}{12b}
\ncline{1a}{1b}
\ncline{1a}{13b}
\ncline{2a}{12b}
\ncline{2a}{23b}
\ncline{2a}{2b}
\ncline{3a}{3b}
\ncline{3a}{13b}
\ncline{3a}{23b}
\ncline{12b}{123c}
\ncline{3b}{123c}
\ncline{1b}{123c}
\ncline{23b}{123c}
\ncline{13b}{123c}
\ncline{2b}{123c}
\endpspicture
\end{center}
\caption{Hasse diagram of $(\mathfrak{C}(N)_\bot,\sqsubseteq)$ with
  $n=3$. Elements with the same partition are framed in grey.}
\label{fig:3}
\end{figure}

As the next proposition will show, $(\C(n)_\bot,\sqsubseteq)$ is a lattice, whose main
properties are given below. We investigate in particular the number of maximal
chains between two elements of this lattice. The reason is that, in game theory,
many notions are defined through maximal chains (e.g., the Shapley value, the
core, etc.).  The cases $n=1,2$ are discarded since trivial
($\C(2)_\bot=2^2$). The standard terminology used hereafter can be found in any
textbook on lattices and posets (e.g., \cite{aig79,dapr90,gratz98}).
\begin{proposition}\label{prop:3}
For any $n>2$, $(\mathfrak{C}(n)_\bot,\sqsubseteq)$ is a lattice, with the following properties:
\begin{enumerate}
\item Supremum and infimum are given by
\begin{align*}
(S,\pi)\vee(S',\pi') & = (T\cup T', \rho) \\
(S,\pi)\wedge(S',\pi') & = (S\cap S', \pi\wedge\pi') \text{ if } S\cap
  S'\neq\emptyset, \text{ and }\bot \text{ otherwise},
\end{align*}
where $T,T'$ are blocks of $\pi\vee \pi'$ containing respectively $S$ and $S'$,
and $\rho$ is the partition obtained by merging $T$ and $T'$ in $\pi\vee\pi'$.
\item Top and bottom elements are $N\{N\}$ and $\bot$. Every element is
  complemented; for a given $S\pi$, any embedded subset of the form
$\overline{S}\pi_{\overline{S}}$ with $\overline{S}$ the complement of $S$ and
$\pi_{\overline{S}}$ any partition containing $\overline{S}$, is a complement of
  $S\pi$.  
\item Each element $S\pi$ where $\pi:=\{S,S_2,\ldots,S_k\}$ is a
  $k$-partition is covered by $\binom{k}{2}$ elements, and covers $\sum_{T\in
  \pi}2^{t-1}-|\pi| + 2^{s-1}-1$ elements.
\item Its join-irreducible elements are $(i,\pi^\bot)$, $i\in N$ (atoms), and
  $(i,\pi_{jk}^\bot)$, $i,j,k\in N$, $i\not\in\{j,k\}$. Its  meet-irreducible
  elements are $(S,\pi)$ where $\pi$ is any 2-partition (co-atoms).
\item The lattice satisfies the Jordan-Dedekind chain condition (otherwise said,
  the lattice is ranked), and its height function is $h(S,\pi)=n-k+1$, if $\pi$
  is a $k$-partition. The height of the lattice is $n$.
\item The lattice is not distributive (and even neither upper nor lower locally
distributive), not atomistic (hence not geometric), not modular but upper semimodular.
\item The number of elements on level of height $k$ is $kS_{n,k}$. The total
  number of elements is $\sum_{k=1}^n kS_{n,k} +1$.
\begin{center}
\begin{tabular}{|c|cccccccc|}\hline
$n$ & 1 & 2 & 3 & 4 & 5 & 6 & 7 & 8 \\ \hline
$|\mathfrak{C}(n)_\bot|$ & 2 & 4 & 11 & 38 & 152 & 675 & 3264 & 17008\\ \hline
\end{tabular}
\end{center}
\item Let $S\pi:=S\{S,S_2,\ldots,S_k\}$ and
$S'\pi':=S'\{S',S_{12},\ldots,S_{1l_1},S_{21},\ldots,
S_{2l_2},\ldots,S_{k1},\ldots,S_{kl_k}\}$, and $k':=\sum_{i=1}^kl_i$. We have
the following isomorphisms:
\begin{align*}
[(i,\pi^\bot),N\{N\}]& \cong\Pi(n)\\
[\bot,S\pi] & \cong  \big(\mathfrak{C}(S)\times\Pi(S_2)\times\cdots\times
\Pi(S_k)\big)\cup\{\bot\}\\ 
[S\pi,N\{N\}] & \cong [i\{i,i_2,\ldots,i_k\}, K\{K\}] \cong \Pi(k)\\
[S'\pi', S\pi] &\cong [i\pi^\bot,S\pi]_{\C(k')_\bot},
\end{align*}
 where the subscript $\C(k')_\bot$ means that elements in the brackets are
understood to belong to $\C(k')_\bot$.
\item The number of maximal chains from $\bot$ to $N\{N\}$ is
  $\frac{(n!)^2}{2^{n-1}}$, which is also the number of maximal longest chains
  in $\mathfrak{C}(n)$.
\begin{center}
\begin{tabular}{|r|cccccccc|}\hline
$n$ & 1 & 2 & 3 & 4 & 5 & 6 & 7 & 8 \\ \hline
$|\mathcal{C}(\mathfrak{C}(n)_\bot)|$ & 1 & 2 & 9 & 72 & 900 & 16 200  & 396 900 & 12 700 800 \\ \hline
\end{tabular}
\end{center}
\item Let $S\pi$ be an embedded subset, with $\pi:=\{S,S_2\ldots,S_k\}$,
  and $|S|=s$. The number of maximal chains from $\bot$ to
  $S\pi$ is
\[
|\mathcal{C}([\bot,S\pi])| =  \frac{s(n-k)!}{2^{n-k}}s!s_2!\cdots s_k!.
\]
The number of maximal chains from $(S,\pi)$ to
  $N\{N\}$ is
\[
|\mathcal{C}([(S,\pi), N\{N\}])| =
 \frac{1}{k}|\mathcal{C}(\mathfrak{C}(k)_\bot)| = \frac{k!(k-1)!}{2^{k-1}}.
\]
\item Let $S'\pi'<S\pi$ with the above notation. The number of maximal chains
  from $S'\pi'$ to $S\pi$ is
\[
|\mathcal{C}([S'\pi',S\pi])| = \frac{l_1(k'-k)}{2^{k'-k}}l_1!l_2!\ldots l_k!
\]
\end{enumerate}
\end{proposition}
\begin{proof}
Consider $(S,\pi),(S',\pi')\in\mathfrak{C}(n)_\bot$. Then $(K,\rho)$ is an upper
bound of both elements iff $K\supseteq S\cup S'$ and $\rho\geq \pi\vee\pi'$
(clearly exists). If $S\cup S'\in\pi\vee\pi'$ then $(S\cup S',\pi\vee\pi')$ is
the least upper bound. If not, since $S\in\pi$ and $S'\in \pi'$ and by
definition of $\pi\vee\pi'$, there exist blocks $T,T'$ of $\pi\vee\pi'$ such
that $T\supseteq S$ and $T'\supseteq S'$. Then $(T\cup T', \rho)$, where $\rho$
is the partition obtained by merging $T$ and $T'$ in $\pi\vee\pi'$, is the least
upper bound of $(S,\pi),(S',\pi')$.

Next, $(S\cap S', \pi\wedge\pi')$ would be the infimum if $S\cap S'$ is a block
of $\pi\wedge\pi'$. If $S\cap S'\neq\emptyset$, then this is the case. If not,
then $\bot$ is the only lower bound. This proves that
$(\mathfrak{C}(n)_\bot,\leq)$ is a lattice, and (i), (ii) hold (the assertion on
complemented elements is clear). 

(iii) Clear from (ii) in Section~\ref{sec:part}. 

(iv) Clear from (iii).

(v) From (iii), $S\pi$ covers $S'\pi'$ implies that if $\pi$ is a
$k$-partition, then $\pi'$ is a $(k+1)$-partition. Hence, a maximal chain from
$S\pi$ to the bottom element has length $n-k+1$, which proves the
Jordan-Dedekind chain condition. Now, the height function is $h(S\pi)=n-k+1$.

(vi) The lattice is not (upper or lower locally) distributive since it contains
diamonds. For example, with $n=3$, the following 5 elements form a diamond (see
Fig.~\ref{fig:3}):
\[
(1,\{1,2,3\}),
(12,\{12,3\}), (1,\{1,23\}), (13,\{13,2\}), (123,\{123\}).
\]
For atomisticity see (iv). Let us prove it is upper semimodular. Since the lattice
is ranked, if $x$ covers $x\wedge y$ and $x\wedge y$, then
both $x$ and $y$ are one level above $x\wedge y$. Hence using (iii), if $x\wedge
y:=(S,\{S,S_2,\ldots,S_k\})$, then $x$ has either the form $(S,\{S,S_i\cup
S_j,\ldots\})$ or $(S\cup S_i,\{S\cup S_i,\ldots\})$, and similarly $y=(S,\{S,S_k\cup
S_l,\ldots\})$ or $(S\cup S_j,\{S\cup S_j,\ldots\})$. To compute $x\vee y$, we
have three cases:
\begin{itemize}
\item $x=(S,\{S,S_i\cup S_j,\ldots\})$ and $y=(S,\{S,S_k\cup S_l,\ldots\})$:
  then $x\vee y = (S,\{S, S_i\cup S_j,S_k\cup S_l,\ldots\})$ if $k,l\neq
  i,j$. If, e.g., $k=i$, $x\vee y = (S,\{S, S_i\cup S_j\cup S_l,\ldots\})$.
\item $x=(S\cup S_i,\{S\cup S_i,\ldots\})$ and $y=(S\cup S_j,\{S\cup
  S_j,\ldots\})$: then $x\vee y = (S\cup S_i\cup S_j,\{S\cup S_i\cup S_j,\ldots\})$.
\item $x=(S,\{S,S_i\cup S_j,\ldots\})$ and $y=(S\cup S_k,\{S\cup
  S_k,\ldots\})$: then $x\vee y=(S\cup S_k,\{S\cup S_k,S_i\cup S_j,\ldots\})$ if
  $k\neq i,j$. If $k=i$, $x\vee y=(S\cup S_i\cup S_j,\{S\cup S_i\cup S_j,\ldots\})$.
\end{itemize}
In all cases, we get a $(k-2)$-partition, so upper modularity holds. Lower
semimodularity  does not hold. Taking the example of $\C(3)_\bot$, $123\{123\}$
covers $12\{12,3\}$ and $3\{3,12\}$, but these elements do not cover
$12\{12,3\}\wedge 3\{12,3\}=\bot$. 

(vii) Clear from the results on $\Pi(n)$.

(viii) Consider the element $(i,\pi^\bot)$ in $\mathfrak{C}(n)_\bot$, $i\in N$. Then
$[(i,\pi^\bot),N\{N\}]$ is a sublattice isomorphic to $\Pi(n)$, since
by (iii) the number of elements covering $(i,\pi^\bot)$ is the same as the
number of elements covering $\pi^\bot$ in $\Pi(n)$, and that this property
remain true for all elements above $(i,\pi^\bot)$.

The other assertions are clear. 

(ix) Since by Prop. \ref{prop:1},
$\mathcal{C}(\Pi(n))=\frac{n((n-1)!)^2}{2^{n-1}}$, and using (viii) and the fact
that there are $n$
mutually incomparable elements  $(i,\pi^\bot)$  in $\mathfrak{C}(n)_\bot$, the
result follows. 

(x) The proof follows the same technique as for Prop. \ref{prop:1} (ii). Using
the second assertion of (viii) and noting that deleting the bottom element does
not change the number of maximal (longest) chains, we can write immediately
\[
|\mathcal{C}([\bot,S\pi])| = \prod_{i=1}^k
 |\mathcal{C}(\Pi(s_i))||\mathcal{C}(\mathfrak{C}(s))|\frac{\Big(\sum_{i=1}^k s_i-1\Big)!}{\prod_{i=1}^k(s_i-1)!}.
\]
The result follows by using Prop. \ref{prop:1} (i)  and (ix).

The second assertion is clear since $[S\pi,N\{N\}]$ is isomorphic to $\Pi(k)$
by (viii).

(xi) Same as for Proposition~\ref{prop:1} (iv).
\end{proof}

\section{Functions on $\C(n)_\bot$}\label{sec:fun}
We investigate properties of some classes of real-valued functions over
$\C(n)_\bot$. As our motivation comes from game theory, we will focus on games,
that is, functions vanishing at the bottom element, and on valuations, which are
related to additive games, another fundamental notion in game theory.
\begin{definition}
A \emph{game in partition function form}  on $N$ (called here for short simply
\emph{game on $\C(N)_\bot$}) is a mapping
$v:\mathfrak{C}(N)_\bot\rightarrow \mathbb{R}$, such that $v(\bot)=0$. The set
of all games in partition function form on $N$ is denoted by $\mathcal{PG}(N)$.
\end{definition}
\begin{definition}
Let $v\in \mathcal{PG}(N)$.
\begin{enumerate}
\item $v$ is \emph{monotone} if $S\pi\sqsubseteq S'\pi'$ implies $v(S\pi)\leq
  v(S'\pi')$. A monotone game on $\C(N)_\bot$ is called a \emph{capacity on
  $\C(N)_\bot$}. A capacity on $\C(N)_\bot$ $v$ is normalized if $v(N\{N\})=1$.
\item $v$ is \emph{supermodular} if for every  $S\pi, S'\pi'$ we have
\[
v(S\pi\vee S'\pi')+v(S\pi\wedge S'\pi') \geq v(S\pi) + v(S'\pi').
\]
It is \emph{submodular} if the reverse inequality holds.
\item A game is \emph{additive} if it is both supermodular and submodular.
\item More generally, for a given $k\geq 2$, a game is \emph{$k$-monotone} if for all
  families of $k$ elements $S_1\pi_1, \ldots,S_k\pi_k$ (not necessarily different), we have
\[
v\big(\bigvee_{i\in K}S_i\pi_i\big) \geq \sum_{J\subseteq
  K,J\neq\emptyset}(-1)^{|J+1|}v\big(\bigwedge_{i\in J}S_i\pi_i\big)
\]
putting $K:=\{1,\ldots,k\}$. A game is \emph{$\infty$-monotone} if it is $k$-monotone
for every $k\geq 2$. Note that $k$-monotonicity implies $k'$-monotonicity for
any $2\leq k'\leq k$.
\item A game is a \emph{belief function} if it is a normalized $\infty$-monotone capacity.
\end{enumerate}
\end{definition}
The following result is due to Barth\'elemy \cite{bar00}.
\begin{proposition}\label{prop:bar}
Let $L$ be a lattice. Then $f$ is monotone and $\infty$-monotone on $L$ if and only if it is
monotone and $(|L|-2)$-monotone. 
\end{proposition}

In lattice theory, a \emph{valuation} (or \emph{2-valuation}) on a lattice $L$
is a real-valued function on $L$ being both super- and submodular
(i.e., it is additive in our terminology). More generally, for a given $k\geq 2$, a
\emph{$k$-valuation} satisfies
\[
v\big(\bigvee_{i\in K}x_i\big) = \sum_{J\subseteq
  K,J\neq\emptyset}(-1)^{|J+1|}v\big(\bigwedge_{i\in J}x_i\big)
\]
for every family of $k$ elements. An $\infty$-valuation is a function $f$ which
is a $k$-valuation for every $k\geq 2$.  The following well-known results
clarify the existence of valuations (see \cite[Ch. X]{bir67}, and also
\cite{bar00}).
\begin{proposition}\label{prop:valu}
Let $L$ be a lattice.
\begin{enumerate}
\item $L$ is modular if and only if it admits a strictly monotone valuation.
\item $L$ is distributive if and only if it admits a strictly monotone
  3-valuation.
\item $L$ is distributive if and only if it is modular and every strictly
  monotone valuation is a $k$-valuation for any $k\geq 2$.  
\item Any lattice admits an $\infty$-valuation.
\end{enumerate}
\end{proposition}
The consequence of (i) is that no strictly monotone additive game exists since
$\C(n)_\bot$ is not modular when $n>2$. The
question is: does it exist an additive game? We shall prove that the answer is
no (except for the trivial game $v=0$) as soon as $n>2$.
\begin{proposition}
$\C(2)_\bot$ admits a strictly monotone 2-valuation (hence by Prop.~\ref{prop:bar}, a
  strictly monotone $\infty$-monotone valuation). If $n>2$, the only possible valuations
  on $\C(n)_\bot$ are constant valuations.
\end{proposition} 
\begin{proof}
For $n=2$, the result is clear from Proposition~\ref{prop:valu} since
$\C(2)_\bot=2^2$. For $n>2$, we shall proceed by induction on $n$. Let us show the
result for $n=3$. Let $f:\C(3)_\bot\rightarrow \mathbb{R}$. To check whether $f$ is a
valuation amounts to verify that
\[
f(x)+f(y) = f(x\vee y) + f(x\wedge y), \quad\forall x,y\in \C(3), \text{ $x$ and
$y$ not comparable}.
\]
This leads to a linear system, for which the constant function is an obvious
solution. Let us show that it is the only one. We extract the following subsystem,
naming for brevity elements $S\pi$ by $a,b,c,\ldots$ as on Figure~\ref{fig:c3}:
\begin{figure}[htb]
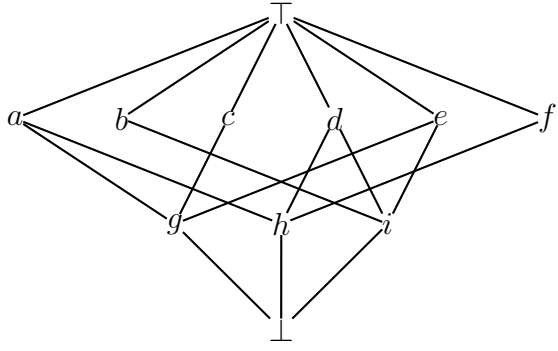

\begin{center}
\psset{unit=0.7cm}
\pspicture(-6,-1)(6,6)
\rput(0,0){\rnode{0}{$\bot$}}
\rput(-2,2){\rnode{1a}{$g$}}
\rput(0,2){\rnode{2a}{$h$}}
\rput(2,2){\rnode{3a}{$i$}}
\rput(-5,4){\rnode{12b}{$a$}}
\rput(-3,4){\rnode{3b}{$b$}}
\rput(-1,4){\rnode{1b}{$c$}}
\rput(1,4){\rnode{23b}{$d$}}
\rput(3,4){\rnode{13b}{$e$}}
\rput(5,4){\rnode{2b}{$f$}}
\rput(0,6){\rnode{123c}{$\top$}}
\ncline{0}{1a}
\ncline{0}{2a}
\ncline{0}{3a}
\ncline{1a}{12b}
\ncline{1a}{1b}
\ncline{1a}{13b}
\ncline{2a}{12b}
\ncline{2a}{23b}
\ncline{2a}{2b}
\ncline{3a}{3b}
\ncline{3a}{13b}
\ncline{3a}{23b}
\ncline{12b}{123c}
\ncline{3b}{123c}
\ncline{1b}{123c}
\ncline{23b}{123c}
\ncline{13b}{123c}
\ncline{2b}{123c}
\endpspicture
\end{center}
\caption{Hasse diagram of $(\mathfrak{C}(3)_\bot,\sqsubseteq)$}
\label{fig:c3}
\end{figure}
\begin{align*}
f(a) + f(b) & = f(\top) + f(\bot)&& (1)\\ 
f(a) + f(c) & = f(\top) + f(g)   && (2)\\ 
f(a) + f(e) & = f(\top) + f(g)   && (3)\\ 
f(a) + f(i) & = f(\top) + f(\bot)&& (4)\\ 
f(b) + f(c) & = f(\top) + f(\bot)&& (5)\\ 
f(b) + f(f) & = f(\top) + f(\bot)&& (6)\\ 
f(b) + f(g) & = f(\top) + f(\bot)&& (7)\\ 
f(b) + f(h) & = f(\top) + f(\bot)&& (8)\\ 
f(c) + f(d) & = f(\top) + f(\bot)&& (9)\\ 
f(c) + f(f) & = f(\top) + f(\bot)&& (10)\\ 
f(c) + f(h) & = f(\top) + f(\bot)&& (11)\\ 
f(f) + f(g) & = f(\top) + f(\bot)&& (12)\\ 
f(f) + f(i) & = f(\top) + f(\bot)&& (13)
\end{align*}
$f(b)=f(i)$ from (1) and (4), $f(c)=f(f)=f(g)=f(h)$ from (5), (6), (7) and
(8), $f(d)=f(f)=f(h)$ from (9), (10) and (11), $f(g)=f(i)$ from (12) and
(13). Then $f(g)=f(\bot)$ from (1) and (2), and hence $f(a)=f(\top)$ from (1),
and $f(e)=f(g)$ from (3). Finally $f(\top)=f(\bot)$ from (5). This proves the
assertion for $n=3$.

Assume the assertion holds till $n$, and let us prove it for $n+1$. The idea is
to cover entirely $\C(n+1)_\bot$ by overlapping copies of $\C(n)_\bot$. Then
since on each copy the valuation has to be constant by assumption, it will be
constant everywhere. Let us consider the set $\C[ij]$ of those embedded
subsets in $\C(n+1)_\bot$ where elements $i,j\in N$ belong to the same block
(i.e., as if $[ij]$ were a single element). Then this set plus the bottom
element $\bot$ of $\C(n+1)_\bot$, which we denote by $\C[ij]_\bot$, is
isomorphic to $\C(n)_\bot$. Moreover, $\C[ij]_\bot$ is a sublattice of
$\C(n+1)_\bot$, since the supremum and infimum in $\C(n+1)_\bot$ of embedded
subsets in $\C[ij]_\bot$ remains in $\C[ij]_\bot$ (because $[ij]$ is never
splitted when taking the infimum and supremum over partitions). This proves that
the system for equations of valuation in $\C[ij]_\bot$ is the same than the set
of equations in $\C(n+1)_\bot$ restricted to $\C[ij]_\bot$. Taking all
possibilities for $i,j$ cover all elements of $\C(n+1)_\bot$, except
atoms. Also, since $\top$ and $\bot$ belongs to each $\C[ij]_\bot$, an overlap
exists. It remains to cover the set of atoms. For this we consider $\C(n)_\bot$,
and to each embedded subset $S\pi$ we add the block $\{n+1\}$ in
$\pi$. Again, this is a sublattice of $\C(n+1)_\bot$ isomorphic to $\C(n)_\bot$,
covering all atoms except $\{n+1\}\pi^\bot$. For this last one, it suffices to
do the same on the set $\{2,\ldots,n+1\}$ and to add the block $\{1\}$ to each
partition.
\end{proof}
Since any game $v$ satisfies $v(\bot)=0$, we have:
\begin{corollary}
The only additive game is the constant
game $v=0$.
\end{corollary}
We comment on this surprising result. Our definition of an additive game follows
tradition in the theory of posets: additivity means both supermodularity and
submodularity, and as a consequence, it converts supremum (similar as union) into
addition, provided the elements do not cover a common element different from the bottom
element (similar to elements with an empty intersection). But it does not match
with the traditional view in game theory, where additive games are assimilated
to imputations, hence to values. An imputation is a vector defined on the set of
players, often indicating how the total worth of the game is shared among
players. Incidentally, in the classical setting, both notions coincide. In this
more complex structure, this is no more the case. The nonexistence of additive
games does not imply therefore the absence of imputation or value. It simply
says that it is not possible to have an additivity property (converting supremum
into a sum) for such a game. The consequence is that games in partition function
form cannot be written in a
simpler form, like an additive game which is equivalent to a $n$-dimensional
vector in the classical setting.   

\section{The M\"obius function on $\C(n)_\bot$}\label{sec:mob}
The M\"obius function is a central notion in combinatorics and posets (see
\cite{rot64}). It is also very useful in cooperative game theory, since it leads
to the M\"obius transform (known in this domain as the Harsanyi dividends of a
game), which are the coordinates of a game in the basis of unanimity games (see
end of this section).

First we give the results for the lattice of partitions $\Pi(n)$.
\begin{proposition}\label{prop:mobpart}
Let $\pi,\sigma$ be partitions in $\Pi(n)$ such that $\pi<\sigma$. Let us denote
by $b(\pi)$ the number of blocks of $\pi$, with $n_1,\ldots,n_{b(\pi)}$ the
sizes of the blocks. Then the M\"obius function on $\Pi(n)$ is given by:
\begin{enumerate}
\item $\mu_{\Pi(n)}(\pi^\bot,\{N\}) = (-1)^{n-1}(n-1)!$
\item $\mu_{\Pi(n)}(\pi,\{N\})= (-1)^{b(\pi)-1}(b(\pi)-1)!$
\item $\mu_{\Pi(n)}(\pi^\bot,\pi) = (-1)^{n-b(\pi)}(n_1-1)!\cdots (n_{b(\pi)-1})!$
\item $\mu_{\Pi(n)}(\pi,\sigma) = (-1)^{b(\pi) -
  b(\sigma)}(m_1-1)!\cdots(m_{b(\sigma)}-1)!$, with $m_1,\ldots,m_{b(\pi)}$
  integers such that $\sum_{i=1}^{b(\sigma)}m_i=b(\pi)$.
\end{enumerate}
\end{proposition}
\begin{proof}
(i) is proved in Aigner \cite[p. 154]{aig79}. The rest is deduced from the
  isomorphisms given in Section~\ref{sec:part}.
\end{proof}

We recall also the following fundamental result \cite{gest95}.
\begin{proposition}\label{th:mob}
If $P$ is a lattice with bottom element 0 and set of atoms $\mathcal{A}$, for
every $x\in P$, the M\"obius function reads
\[
\mu(0,x) = \sum_{S\subseteq \mathcal{A}\mid \bigvee S=x}(-1)^{|S|}.
\] 
\end{proposition}
In $\C(N)_\bot$, there are only a few elements representable by atoms. These are
$S\pi^\bot_S$, where $\pi^\bot_S$ is the partition formed by $S$ and
singletons. The unique way to write $S\pi^\bot_S$  is $\bigvee_{i\in
  S}(i\pi^\bot)$. Hence, the M\"obius function over $\C(N)_\bot$, denoted simply
by $\mu$ if no confusion occurs, otherwise by $\mu_{\C(N)_\bot}$, is known for
$(\bot,S\pi)$: 
\begin{equation}\label{eq:mob}
\mu(\bot,S\pi) = \begin{cases}
  (-1)^{|S|}, & \text{if } \pi=\pi^\bot_S\\
  0, & \text{otherwise.}
  \end{cases}
\end{equation}

Let us find $\mu(S\pi,N\{N\})$. Since $[S\{S,S_2,\ldots,S_k\},N\{N\}]\cong
[i\{i,i_2,\ldots,i_k\},K\{K\}]$ (see Proposition~\ref{prop:3} (viii)), we have
$\mu_{\C(n)_\bot}(S\{S,S_2,\ldots,S_k\},N\{N\}) =
\mu_{\C(k)_\bot}(i\pi^\bot,K\{K\})$. We know also that $[i\pi^\bot,N\{N\}]\cong
\Pi(N)$, from which we deduce by Proposition~\ref{prop:mobpart}
that $\mu(i\pi^\bot,N\{N\}) = \mu_{\Pi(n)}(\pi^\bot,\{N\}) = (-1)^{n-1}(n-1)!$. Hence
\begin{equation}\label{eq:mob2}
\mu(S\{S,S_2,\ldots,S_k\},N\{N\}) = (-1)^{k-1}(k-1)!.
\end{equation}

It remains to compute the general expression for $(S'\pi',S\pi)$. We put
$S\pi:=S\{S,S_2,\ldots,S_k\}$ and
$S'\pi':=S'\{S',S_{12},\ldots,S_{1l_1},S_{21},\ldots,
S_{2l_2},\ldots,S_{k1},\ldots,S_{kl_k}\}$, and $k':=\sum_{i=1}^kl_i$. Then
$[S'\pi',S\pi] \cong [i\pi^\bot,S\pi]_{\C(k')_\bot}$, where the subscript
$\C(k')_\bot$ means that elements in the brackets are understood to belong to
$\C(k')_\bot$. Hence we deduce:
\[
\mu(S'\pi',S\pi) = \mu_{\Pi(k')}(\pi^\bot,\pi) = (-1)^{k'- k}(l_1-1)!\cdots(l_k-1)!.
\] 
In summary, we have proved:
\begin{proposition}\label{prop:mobemb}
The M\"obius function on $\C(n)_\bot$ is given by (with the above notation):
\begin{align}
\mu(\bot,S\pi) & = \begin{cases}
  (-1)^{|S|}, & \text{if } \pi=\pi^\bot_S\\
  0, & \text{otherwise.}
  \end{cases}\\
\mu(S'\pi',S\pi) & = (-1)^{k'- k}(l_1-1)!\cdots(l_k-1)!, \text{ for } S'\pi'\sqsubseteq S\pi.
\end{align}
\end{proposition}
In particular, $\mu(i\pi^\bot,S\pi) = (-1)^{n-k} (s-1)!(s_2-1)!\cdots(s_k-1)!$.

Using this, the M\"obius transform of any game $v$ on $\C(N)_\bot$ is defined by 
\[
m(S\pi) = \sum_{S'\pi'\sqsubseteq S\pi}\mu(S'\pi',S\pi)v(S'\pi'),\quad\forall S\pi\in\C(N)_\bot.
\] 
Now, $m$ gives the coordinates of $v$ in the basis of unanimity games
defined by:
\[
u_{S\pi}(S'\pi') = \begin{cases}
  1, & \text{if } S'\pi'\sqsupseteq S\pi\\
  0, & \text{ otherwise.}
  \end{cases}
\]
\begin{example}\label{ex:mob3}
Application for $n=3$. Let us compute the M\"obius transform of a game $v$. We
have, for all distinct $i,j,k\in\{1,2,3\}$:
\begin{align*}
m(\bot) & = 0\\
m(i\pi^\bot) & = v(i\pi^\bot)\\
m(ij\{ij,k\}) & = v(ij\{ij,k\}) - v(i\pi^\bot) - v(j\pi^\bot)\\
m(i\{i,jk\}) & = v(i\{i,jk\}) - v(i\pi^\bot)\\
m(123\{123\}) & = v(123\{123\}) - \sum_{i,j}v(ij\{ij,k\}) -
\sum_{i}v(i\{i,jk\}) + 2\sum_{i}v(i\pi^\bot).
\end{align*}
Consider now the following game: $v(123\{123\})=3$, $v(12\{12,3\})=2$,
$v(3\{12,3\})=0$, $v(1\{1,23\})=1$, $v(23\{1,23\})=2$, $v(13\{13,2\})=1$,
$v(2\{13,2\})=1$, $v(2\{13,2\})=1$, $v(1\{1,2,3\})=1$, $v(2\{1,2,3\})=1$, and
$v(3\{1,2,3\})=0$. Let us find its coordinates in the basis of unanimity
games. It suffices to compute $m(S\pi)$ for all $S\pi\in\C(3)$ from the above
formulas. We obtain:
\[
m(1\{1,2,3\})=m(2\{1,2,3\}) = m(23\{1,23\})=1
\]   
and $m(S\pi)=0$ otherwise. Therefore, 
\[
v= u_{1\{1,2,3\}} + u_{2\{1,2,3\}} + u_{23\{1,23\}}.
\]
This decomposition can be checked on the figure of $\C(3)$ (Figure~\ref{fig:3}). 
\end{example}

\section{Belief functions and minitive functions on $\C(n)_\bot$}\label{sec:bel}
Belief functions and minitive functions (i.e., inf-preserving mappings, also
called necessity measures) are well-known in artificial intelligence and
decision making, where the underlying lattice is the Boolean lattice. In the
case of an arbitrary lattice, they have interesting properties, investigated by
Barthélemy \cite{bar00} and the author \cite{gra05}.

Barth\'elemy proved the following.
\begin{proposition}\label{prop:bar1}
Let $L$ be any lattice, and $f:L\rightarrow \mathbb{R}$. If the M\"obius
transform of $f$, denoted by $m$, satisfies $m(\bot)=0$, $m(x)\geq 0$  for all
$x\in L$, and
$\sum_{x\in L}m(x)=1$ (normalization), then $f$ is a belief function.
\end{proposition}
The converse of this proposition does not hold in general (it holds for the
Boolean lattice $2^N$). A belief function is \emph{invertible} if its M\"obius
transform is nonnegative, normalized and vanishes at $\bot$. The following counterexample
shows that for $\C(n)_\bot$, there exist belief functions which are not
invertible. 
\begin{example}
Let us take $\C(3)_\bot$, and consider a function $f$ whose values are given on
the figure below.
\begin{figure}[htb]
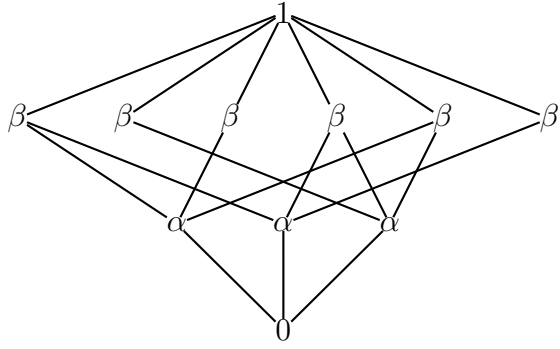

\begin{center}
\psset{unit=0.7cm}
\pspicture(-6,-1)(6,6)
\rput(0,0){\rnode{0}{$0$}}
\rput(-2,2){\rnode{1a}{$\alpha$}}
\rput(0,2){\rnode{2a}{$\alpha$}}
\rput(2,2){\rnode{3a}{$\alpha$}}
\rput(-5,4){\rnode{12b}{$\beta$}}
\rput(-3,4){\rnode{3b}{$\beta$}}
\rput(-1,4){\rnode{1b}{$\beta$}}
\rput(1,4){\rnode{23b}{$\beta$}}
\rput(3,4){\rnode{13b}{$\beta$}}
\rput(5,4){\rnode{2b}{$\beta$}}
\rput(0,6){\rnode{123c}{$1$}}
\ncline{0}{1a}
\ncline{0}{2a}
\ncline{0}{3a}
\ncline{1a}{12b}
\ncline{1a}{1b}
\ncline{1a}{13b}
\ncline{2a}{12b}
\ncline{2a}{23b}
\ncline{2a}{2b}
\ncline{3a}{3b}
\ncline{3a}{13b}
\ncline{3a}{23b}
\ncline{12b}{123c}
\ncline{3b}{123c}
\ncline{1b}{123c}
\ncline{23b}{123c}
\ncline{13b}{123c}
\ncline{2b}{123c}
\endpspicture
\end{center}
\caption{Hasse diagram of $(\mathfrak{C}(3)_\bot,\sqsubseteq)$}
\label{fig:3b}
\end{figure}
Monotonicity implies that $1\geq\beta\geq \alpha\geq 0$. In order to check
$\infty$-monotonicity, from Proposition~\ref{prop:bar} we know that it suffices
to check till 7-monotonicity. We write below the most constraining inequalities
only, keeping in mind that $1\geq\beta\geq \alpha\geq 0$.

2-monotonicity is equivalent to  $\beta\geq 2\alpha$ and $1\geq 2\beta-\alpha$. 
3-monotonicity is equivalent to $1\geq 3\beta - \alpha$. 4-monotonicity is
equivalent to $1\geq 4\beta-3\alpha$. 5-monotonicity is equivalent to $1\geq
5\beta - 4\alpha$, while 6-monotonicity is equivalent to $1\geq
6\beta-9\alpha$. 7-monotonicity does not add further constraints. 

From Example~\ref{ex:mob3}, nonnegativity of the M\"obius transform implies
$\alpha\geq 0$ (atoms), $\beta\geq 2\alpha$ (2nd level), and $m(\top)=6\alpha -
6\beta +1\geq 0$. Then taking $\alpha=0.1$, $\beta=0.28$ make that $f$ is a
belief function, but $m(\top)=-0.08$.
\end{example}  

The last point concerns minitive functions. A \emph{minitive function} on a
lattice $L$ is a real-valued function $f$ on $L$ such that $f(\bot)=0$,
$f(\top)=1$, and
\[
f(x\wedge y) = \min(f(x),f(y)).
\]
Hence a minitive function is a capacity.
The following proposition by Barth\'elemy shows that minitive functions always
exist on $\C(n)_\bot$.
\begin{proposition}
Let $L$ be a lattice and $f:L\rightarrow\mathbb{R}$ such that $f(\bot)=0$,
$f(\top)=1$. Then $f$ is a minitive function if and only if it is an invertible
belief function whose M\"obius transform is nonzero on a chain of $\C(n)_\bot$.  
\end{proposition} 
In other words, taking any chain $C$ in $\C(n)$ and assigning nonnegative
numbers on elements of $C$ such that their sum is 1 generates (by
Proposition~\ref{prop:bar1}) a belief function which is a minitive function.

\section{Concluding remarks}
The paper has given a natural structure to embedded subsets, and hence
a better understanding to games
in partition function form. Specifically, the main results from a game theoretic
viewpoint are:
\begin{itemize}
\item The set of embedded subsets forms a lattice, whose structure is much more
  complicated and less easy to handle than the Boolean lattice of coalitions in the classical
  setting. In particular, since the lattice is not distributive, no simple
  decomposition of elements is possible, a fortiori no decompositions in atoms. 
\item The number of maximal chains between any two elements is known. This
  allows the definition of many notions in game theory, like the Shapley value. 
\item The decomposition of games into the
  basis of unanimity games, and hence its Harsanyi dividends, is known.
\item There is no additivity property for such games, hence there is few hope to
  express them in a simpler form. 
\item Infinite monotonicity is no more equivalent to the nonnegativity of the
  M\"obius transform (Harsanyi dividends). 
\end{itemize}
We hope that this work can help in the clarification and  further investigation
on games in partition function form.

\bibliographystyle{plain}

\bibliography{../BIB/fuzzy,../BIB/grabisch,../BIB/general}

\end{document}